\documentclass[11pt]{article}

\usepackage{amsmath,amsthm,amssymb}
\usepackage{fullpage}
\usepackage{hyperref}
\usepackage{authblk}
\usepackage{amssymb} 
\usepackage{verbatim} 
\usepackage{footnote} 
\usepackage{xspace}

\newtheorem{theorem}{Theorem}
\newtheorem{corollary}[theorem]{Corollary}
\newtheorem{lemma}{Lemma}
\theoremstyle{definition}
\newtheorem{definition}{Definition}
\newcommand{\prob}[1]{\textsc{\lowercase{#1}}\xspace}
\newcommand{\Oh}{\mathcal{O}} 
\newcommand{\NP}{\ensuremath{\mathsf{NP}}\xspace}

\DeclareMathOperator{\tw}{\mathsf{tw}}
\DeclareMathOperator{\td}{\mathsf{td}}
\DeclareMathOperator{\cw}{\mathsf{cw}}
\DeclareMathOperator{\mw}{\mathsf{mw}}

\DeclareMathOperator*{\argmin}{arg\,min}

\date{}
\makeatletter 
\renewcommand\subparagraph{\@startsection{subparagraph}{5}{0pt}%
                                       {3.25ex \@plus1ex \@minus .2ex}%
                                       {-1em}%
                                      {\normalfont\normalsize\bfseries}} 
\makeatother 

\title{Efficient and adaptive parameterized algorithms on modular decompositions}

\author[1]{Stefan Kratsch}
\author[1]{Florian Nelles}

\affil[1]{\small Department of Computer Science, Humboldt-Universit{\"a}t zu
Berlin, Germany 
{\{kratsch,nelles\}@informatik.hu-berlin.de}}

\begin{document}

\maketitle

\begin{abstract}
We study the influence of a graph parameter called modular-width on the time complexity for optimally solving well-known polynomial problems such as \prob{Maximum Matching}, \prob{Triangle Counting}, and \prob{Maximum $s$-$t$ Vertex-Capacitated Flow}. The modular-width of a graph depends on its (unique) modular decomposition tree, and can be computed in linear time $\Oh(n+m)$ for graphs with $n$ vertices and $m$ edges. Modular decompositions are an important tool for graph algorithms, e.g., for linear-time recognition of certain graph classes.

Throughout, we obtain efficient parameterized algorithms of running times $\Oh(f(\mw)n+m)$, $\Oh(n+f(\mw)m)$ , or $\Oh(f(\mw)+n+m)$ for graphs of modular-width $\mw$. Our algorithm for \prob{Maximum Matching}, running in time $\Oh(\mw^2\log \mw n+m)$, is both faster and simpler than the recent $\Oh(\mw^4n+m)$ time algorithm of Coudert et al.\ (SODA 2018). For several other problems, e.g., \prob{Triangle Counting} and \prob{Maximum $b$-Matching}, we give adaptive algorithms, meaning that their running times match the best unparameterized algorithms for worst-case modular-width of $\mw=\Theta(n)$ and they outperform them already for $\mw=o(n)$, until reaching linear time for $\mw=\Oh(1)$.
\end{abstract}

\section{Introduction}

Determining the best possible worst-case running times for computational problems lies at the heart of algorithmic research. For many intensively studied problems progress has been stalled for decades and one may suspect that the ``correct'' running times have already been found. While there is still only little known regarding unconditional lower bounds, the recent success of ``fine-grained analysis of algorithms'' has brought plenty of tight conditional lower bounds for a wealth of problems (see, e.g., \cite{PatrascuW10,Bringmann14,AbboudWY15}). Indeed, if one is willing to believe in the conjectured worst-case optimality of known algorithms for \prob{3-SUM}, \prob{All-Pairs-Shortest Paths} (APSP), or \prob{Satisfiability}\footnote{It has been conjectured that there is no $\Oh(n^{2-\varepsilon})$ time algorithm for 3-SUM, no $\Oh(n^{3-\varepsilon})$ time for APSP, and there is no $c<2$ such that $k$-SAT can be solved in time $\Oh(c^n)$ for each fixed $k$ (SETH).} then lots of other known algorithms must be optimal as well. Even if there is no general agreement on the truth of the conjectures, the previously stalled work can now be focused on beating the best known times for just those problems rather than for a multitude of problems.

Complementary to the quest for refuting conjectures and beating long-standing fastest algorithms, what should we do if the conjectures and implied lower bounds are true (or if we simply fail to disprove them)? Certainly, quadratic or cubic time is often too slow, even long before entering the realm of big data. Apart from heuristics and approximate algorithms, a possible solution lies in taking advantage of structure in the input and deriving worst-case running times that depend on parameters that quantify this structure. Consider for example the \prob{Longest Common Subsequence} problem, where a breakthrough result \cite{AbboudBW15,BringmannK15} proved that there is no $\Oh(n^{2-\varepsilon})$ time algorithm for any $\varepsilon>0$ unless \prob{Satisfiability} can be solved in $\Oh((2-\varepsilon')^n)$ time for some $\varepsilon'>0$ and SETH fails. Long before this result, algorithms were discovered that run much faster than $\Oh(n^2)$ time when certain parameters are small (cf.~\cite{BringmannK18}); curiously, a very recent result of Bringmann and Künnemann~\cite{BringmannK18} shows that these are optimal modulo SETH (while giving one new optimal algorithm for binary alphabets). 
Similarly, for the task of sorting an array of $n$ items, there is the (unconditional) lower bound of $\Omega(n \log n)$ for comparison-based sorting, which is matched by well-known sorting algorithms. The goal in the area of adaptive sorting is to find algorithms that are adaptive to presortedness (a.k.a., input structure) with very low running times for almost sorted inputs while maintaining competitive running times as disorder increases (cf.~\cite{Estivill-CastroW92}).

The success of fine-grained analysis has rekindled the interest in outperforming (possibly optimal) worst-case running times by tailoring algorithms to benefit from input structure. This fits naturally into the framework of parameterized complexity where running times are expressed in terms of input size and one or more problem-specific parameters. Usually, this is aimed at \NP-hard problems and a key goal is to obtain \emph{fixed-parameter tractable} (FPT) algorithms that run in time $f(k)n^c$ where $f(k)$ is a (usually exponential) function of the parameter and $n^c$ denotes a fixed polynomial in the input size $n$. Recent work of Giannopoulou et al.~\cite{GiannopoulouMN15} has initiated a programmatic study of what they called ``FPT in P'', i.e., efficient parameterized algorithms for tractable problems. Here, they propose to seek running time $\Oh(k^\alpha n^\beta)$ when the best dependence on input size alone is $\Oh(n^\gamma)$ for $\gamma>\beta$; in particular, algorithms with linear dependence on the input size are sought, i.e., time $\Oh(k^\alpha n)$. Giannopoulou et al.\ suggest that \prob{Maximum Matching} could become a focal point of study, similar to the related \NP-hard \prob{Vertex Cover} problem in parameterized complexity.

There have been several recent publications that fit into the FPT in P program~\cite{fomin2017fully,mertzios2016fine,bentert2017parameterized,fluschnik2017can,iwata2017power}. 
Several works focus on the \emph{treewidth} parameter, which is of core importance in parameterized complexity \cite{fomin2017fully,Husfeldt16}.
In particular, Fomin et al.~\cite{fomin2017fully} obtained algorithms that depend polynomially on input size $n$ and treewidth $\tw$ to solve a number of problems related to determinants and systems of linear inequalities; e.g., they can solve \prob{Maximum Matching} in time $\Oh(\tw^3n\log n)$ and vertex flow with unit capacities in time $\Oh(\tw^2n \log n)$. (A small caveat of treewidth in this context is that it is \NP-hard to compute so one has to resort to an approximation with polynomial blow-up in the treewidth.) Iwata et al.~\cite{iwata2017power} studied the related parameter \emph{tree-depth} and, among other results, showed how to solve \prob{Maximum Matching} in time $\Oh(\td m)$ on graphs of tree-depth $\td$. Very recently, Coudert et al.~\cite{CoudertDP18} studied another tree-width related parameter called \emph{clique-width} as well as several related parameters such as modular-width and split-width; they obtain upper and lower bounds for a variety of problems. Their main result is an algorithm for \prob{Maximum Matching} that runs in $\Oh(\mw^4 n + m)$ time, where $\mw$ stands for the modular-width of the input graph. Note that modular-width and the modular decomposition of a graph can be computed in linear time $\Oh(n+m)$; the modular-width is an upper bound for the (\NP-hard) clique-width but it is itself unbounded already on graphs of constant clique-width.

\subparagraph{Our work.}
\begin{savenotes}
\begin{table}
\begin{tabular}{l|l|l}
\multicolumn{1}{c|}{\bfseries Problem} & \multicolumn{1}{c|}{\bfseries Best unparameterized} & \multicolumn{1}{c}{\bfseries Our result}  \\ \hline
  
  \prob{Maximum Matching} &  $\Oh(m \sqrt{n})$ \cite{MicaliV80}  & $\Oh(\mw^2 \log \mw n +m)$ \\ \hline
  \prob{Maximum $b$-Matching}\footnote{For $b(V) \geq n \log n$} & $\Oh( (n \log n) \cdot (m + n \log n ))$ \cite{Gabow16a}  & $\Oh(\mw^2 \log \mw n +m)$ or  \\
  &&$\Oh((\mw \log \mw)\cdot(m+n \log \mw))$\\ \hline
  \prob{Triange Counting} & $\Oh(n^{\omega})$ \cite{schank2005finding} or & $\Oh(\mw^{\omega-1}n+m)$ \\
  &$\Oh(m^{\frac{2\omega}{\omega+1}})=\Oh(m^{1.41})$ \cite{alon1997finding}&\\ \hline
  \prob{Edge-Disjoint $s$-$t$ Paths} & $\Oh(n^{\frac{3}{2}} m^{\frac{1}{2}})$ \cite{goldberg1999flows}&  $\Oh({\mw^3}+n+m)$\\ \hline
  \prob{Global Min Cut} & $\Oh(m + \lambda^2n\log(n/\lambda))$ \cite{gabow1995matroid} &  $\Oh({\mw^3}+n+m)$ \\ \hline
  \prob{Max $s$-$t$ Vertex Flow} & $\Oh(nm)$ \cite{orlin2013max} &  $\Oh({\mw^3} + n + m)$ \\ \hline
  \prob{Global Vertex Min Cut} & $\Oh(n^3 \log n)$ \cite{hao1994faster}&$\Oh(\mw^2\log\mw n+m)$ \\
 \end{tabular}
 \caption{Overview about our results. $n$ and $m$ denote the number of vertices and edges, $\mw$ denotes the modular-width of the input graph, and $\lambda$ denotes the edge-connectivity of the graph (which is upper-bounded by the minimum degree $\delta$, so $\lambda \leq \delta \leq 2m/n$. The previous best result for \prob{Maximum Matching}, parameterized by modular-width $\mw$, was $\Oh(\mw^4n+m)$ \cite{CoudertDP18}. }
\end{table}\label{table:overview}
\end{savenotes}
We further explore the algorithmic applications of modular-width for well-studied tractable problems. See Table~\ref{table:overview} for an overview of our results. First, we improve the running time for \prob{Maximum Matching} from $\Oh(\mw^4n+m)$ to $\Oh(\mw^2 \log \mw n +m)$. We follow the same natural recursive approach as in previous work, i.e., computing optimal solutions in a bottom-up fashion on the modular decomposition tree. Unlike Coudert et al.~\cite{CoudertDP18}, however, we do not seek to use the structure of modules to speed up the computation of augmenting paths, starting from an union of maximum matchings for the child modules. Instead, we simplify the current graph, while retaining the same maximum matching size, such that the  found solutions can be encoded into vertex capacities in a graph with at most $3 \mw$ vertices. This allows us to forget the matchings for the modules and instead of augmenting paths it suffices to find a maximum $b$-matching subject to vertex capacities; using an $\Oh(\min \{ b(V), n \log n \} \cdot (m + n \log n ))=\Oh(n^3\log n)$ time algorithm due to Gabow~\cite{Gabow16a} then yields the claimed running time.\footnote{The obvious upper bound of $\Oh(\mw^3\log\mw n+m)$ of applying Gabow's algorithm on each prime node can be improved by a slightly more careful summation; the same applies in the other results.} 

Our algorithm for \prob{Maximum Matching} easily generalizes to computing maximum $b$-matchings in the same time $\Oh(\mw^2 \log \mw n +m)$. 
By a different summation of the running time, one can also bound the time by $\Oh((\mw \log \mw)\cdot(m+n \log \mw))$. For large total capacity $b(V)$, Gabow's algorithm runs in time $\Oh((n\log n)\cdot (m+n \log n))$, which matches our running time for graphs with worst-case modular-width of $\mw=\Theta(n)$.

Thus, when capacities are large, our algorithm interpolates smoothly between linear time $\Oh(n+m)$ for $\mw=\Oh(1)$ and the running time of the best unparameterized algorithm for $\mw=\Theta(n)$; i.e., it is an \emph{adaptive algorithm} and already $\mw=o(n)$ gives an improved running time. Such adaptive algorithms (for other problems and parameter) were also considered by Iwata et al. \cite{iwata2017power}. For \prob{Maximum Matching}, the comparison with the $\Oh(m\sqrt{n})$ time algorithm of Micali and Vazirani~\cite{MicaliV80} is of course less favorable, but still yields a fairly large regime for $\mw$ where we get a faster algorithm. 

We next study \prob{Triangle Counting} where, given a graph $G=(V,E)$, we need to determine the number of triangles in $G$. The fastest known algorithm in terms of $n$ relies on fast matrix multiplication and runs in $\Oh(n^{\omega})$ time \cite{schank2005finding} where $\omega$ is the matrix multiplication exponent.\footnote{It is known that $2\leq \omega < 2.3728639$ due to Le Gall~\cite{Gall14a}. By definition of $\omega$ the running time is in fact $\Oh(n^{\omega+o(1)})$; adopting a common abuse of notation we use exponent $\omega$ for brevity.} We present an algorithm that runs in $\Oh(\mw^{\omega-1}n+m)$ time. Again, our running time smoothly interpolates between linear time $\Oh(n+m)$ for $\mw=\Oh(1)$ and the best unparameterized time for $\mw=\Theta(n)$, making it adaptive for sufficiently dense graphs; else, the $\Oh(m^{\frac{2\omega}{\omega+1}})=\Oh(m^{1.41})$ time algorithm of Alon et al.~\cite{alon1997finding} is faster. Coudert et al.~\cite{CoudertDP18} obtained time $\Oh(\cw^2(n+m))$ where $\cw$ is the clique-width of $G$; this is incomparable with our result because clique-width is a smaller parameter ($\cw\leq\mw$ and there are graphs with $\cw=\Oh(1)$ but $\mw=\Theta(n)$) but (so far) allows a worse time.

Finally, we turn to problems related to edge- and vertex-disjoint paths. Our results for the latter generalize to vertex-capacitated flows and global min cuts; it is easy to see that there is little use for modular-width for most edge-weighted/capacitated problems because it suffices to solve them on cliques, which have modular-width equal to two (see also Section~\ref{sec:conclusion}). Note that standard transformations between different variants of path- and flow-type problems do not apply here because they affect the modular-width of the graph.
We obtain the following running times: \prob{Maximum $s$-$t$ Vertex-Capacitated Flow} in $\Oh({\mw^3} + n + m)$ time; \prob{Global Vertex-Capacitated Min Cut} in $\Oh(\mw^2\log\mw n+m)$ time; \prob{Edge-Disjoint $s$-$t$ Paths} in $\Oh({\mw^3}+n+m)$ time; and \prob{Unweighted Global Min Cut} in $\Oh({\mw^3}+n+m)$ time. The running times for flows/paths are linear in the graph size and only have an additive contribution in terms of the modular-width, because at most one involved computation (on a prime node) is needed. These also give rise to linear-time kernelization-like algorithms that return an equivalent instance of size $poly(\mw)$, which is the one instance that one would run some other algorithm on (i.e., the only source of non-linear time). Such results (for other problems) have also been observed by Coudert et al.~\cite{CoudertDP18}. It is easy to see that \emph{any} algorithm of running time $\Oh(f(k)+n+m)$, for some parameter $k$, implies a linear-time kernelization: Run the algorithm for $c(n+m)$ steps, for sufficiently large $c$ relative to hidden constants in $\Oh$; it either terminates and returns the correct answer or allows the conclusion that $n+m<f(k)$, i.e., the input instance itself is the kernel.
Again, as done for \prob{Maximum $b$-Matching}, one can obtain different bounds for the running time by slightly different summations. For example, the running time for \prob{Maximum $s$-$t$ Vertex-Capacitated Flow} can also be bounded by $\Oh(\mw m +n)$, meaning that the algorithm is never worse than the optimal unparameterized algorithm and outperforms it already for $\mw=o(n)$.

To summarize, we obtain several results that fit into the recent FPT in P program (and the much older programs of adaptive algorithms and faster algorithms for restricted settings), i.e., efficient parameterized algorithms with running times $\Oh(poly(\mw)(n+m))$ or $\Oh(poly(\mw)+m+n)$. All running times are linear for $\mw=\Oh(1)$ and several algorithms are adaptive so that they match the best known algorithm for $\mw=\Theta(n)$ and outperform it already when $\mw=o(n)$, possibly only for sufficiently dense graphs. Of course, we use the best algorithms as black boxes so the message is that throughout there is little to no overhead even in the worst case for using a modular decomposition-based approach and getting savings in running time already for large (but not worst-case) modular-width.

\subparagraph{Related work.}
\prob{Triangle Counting} is solvable in time $\Oh(n^\omega)$ using fast matrix multiplication~\cite{alon1997finding}, and even for the simpler \prob{Triangle Detection} problem, where only (non-)existence of a single triangle needs to be reported, it has been conjectured that there is no $\Oh(n^{\omega-\varepsilon})$ time and no combinatorial $\Oh(n^{3-\varepsilon})$ time algorithm. The fastest known algorithm for counting triangles in sparse graphs is the AYZ algorithm due to Alon, Yuster, and Zwick~\cite{alon1997finding}, which runs in time $\Oh(m^{\frac{2\omega}{\omega+1}})$ ($\Oh(m^{1.41})$ for $\omega < 2.373$). 
Coudert et al.~\cite{CoudertDP18} gave a faster algorithm for graphs of bounded clique-width $\cw$, running in time $\Oh(\cw^2 (n + m))$. Bentert et al.~\cite{bentert2017parameterized} have studied \prob{Triangle Enumeration} under various parameters including feedback edge number, distance to d-degenerate graphs, and clique-width. The latter one outputs all triangles in time $\Oh(\cw^2n + n^2 + \#T)$ where $\#T$ denotes the number of triangles in $G$.

 The currently best maximum flow algorithm is due to Orlin~\cite{orlin2013max} and runs in time $\Oh(nm)$. Using a flow algorithm, one can determine the number of edge- or vertex-disjoint $s$-$t$ paths in a graph, but in the unweighted case one can do slightly better, e.g., computing the number of edge-disjoint paths in an undirected graph can be done in time $\Oh(n^{\frac{3}{2}} m^{\frac{1}{2}})$ using an algorithm due to Goldberg and Rao~\cite{goldberg1999flows}. Finding a global minimum edge cut with weights on the edges in an undirected graph can be done in time $\Oh(nm+n^2\log n)$ due to Stoer and Wagner~\cite{stoer1997simple}. The unweighted variant can be solved in time $\Oh(m + \lambda^2n\log(n/\lambda))$ by Gabow~\cite{gabow1995matroid}, where $\lambda$ denotes the edge-connectivity of the graph (which is upper-bounded by the minimum degree $\delta$, so $\lambda \leq \delta \leq 2m/n$). There is also a randomized algorithm with running time $\Oh(m\log^3 n)$ due to Karger~\cite{karger2000minimum}.

 The notion of a modular decomposition was first introduced by Gallai~\cite{gallai1967transitiv} for recognizing comparability graphs. The first linear time algorithm to compute a modular decomposition was independently developed by McConnell and Spinrad~\cite{mcconnell1994linear} and Cournier and Habib~\cite{cournier1994new}. Tedder et al.~\cite{tedder2008simpler} later gave a new and much simpler linear-time algorithm.

\subparagraph{Organization.}
In Section~\ref{sec:preliminaries} we briefly introduce basic notation, define the modular decomposition tree, and define modular-width. Then, in Section~\ref{sec:maximummatching}, we consider the problem \prob{Maximum Matching} and the generalization to \prob{Maximum $b$-Matching}. In Section~\ref{sec:trianglecounting}, we study the problem \prob{Triangle Counting}. The remaining results for edge/vertex-disjoint paths, flows, and cuts are discussed in Sections \ref{sec:edgedisjointpaths} and \ref{sec:vertexdisjointpaths}. We conclude in Section \ref{sec:conclusion}.

\section{Preliminaries}\label{sec:preliminaries}

We use standard graph notation~\cite{Diestel12}. An \emph{$s$-$t$ vertex-capacitated flow} in a graph $G=(V,E)$ with vertex capacities $c\colon V\to\mathbb{R}$ is a weighted collection of $s$-$t$ paths in $G$ such that the total weight of paths including any vertex $v\in V\setminus\{s,t\}$ is at most the capacity $c(v)$. (Equivalently, one may define this as a function $f\colon E(\overleftrightarrow{G})\to\mathbb{R}$ where $\overleftrightarrow{G} = (V,A)$ with $A=\{(u,v),(v,u)\mid\{u,v\} \in E\}$ that has flow-conservation at each $v\in V\setminus\{s,t\}$ and with $\sum_{(u,v)\in\delta^{-}_{\overleftrightarrow{G}}(v)} f((u,v))\leq c(v)$ for all $v\in V\setminus\{s,t\}$, where $\delta^{-}_{\overleftrightarrow{G}}(v)$ is the set of arcs with end in $v$.) The value of such a flow, denoted by $\vert f \vert$, is the total weight over all the $s$-$t$ paths (equivalently, $\sum_{(v,t)\in\delta^{-}_{\overleftrightarrow{G}}(t)}f(v,t)$). For unit capacities $c\equiv 1$ this is equivalent to a maximum collection of vertex-disjoint $s$-$t$ paths.

We say that two sets $A$ and $B$ \emph{overlap} if $A \cap B \neq \emptyset$, $A \setminus B \neq \emptyset$,
and $B \setminus A \neq \emptyset$ and let $[n] = \{1, 2, \ldots, n\}$ for any $n \in \mathbb{N}$.

\subparagraph{Modular Decomposition.} 
Let $G = (V, E)$ be a graph. A \emph{module} is a vertex set $M \subseteq V$ such that all vertices in $M$
have the same neighborhood in $V \setminus M$. In other words, $M \subseteq N(x)$ or 
$ M \cap N(x) = \emptyset$ for every vertex $x \in V \setminus M$. Clearly, $\emptyset$, $V$, and $\{ v \}$ 
for every $v \in V$ are modules of $G$; these are called \emph{trivial modules}. If a graph only admits trivial modules, we call $G$ \emph{prime}.  
Consider a partition $P = \{M_1, M_2, \ldots, M_\ell\}$ of the vertices of $G$ into modules where $\ell \geq 2$, called \emph{modular partition}. If there is $v \in M_i$ and $u \in M_j$ with $\{u, v\} \in E$, then any vertex in $M_i$ is adjacent to every vertex in $M_j$. In this case, we can call two modules $M_i$ and $M_j$ of $P$ \emph{adjacent}, and \emph{non-adjacent} otherwise. 

\begin{definition}
Let $P = \{M_1, M_2, \ldots, M_\ell\}$ be a modular partition of a graph $G = (V, E)$. The \emph{quotient graph} $G_{/P} = (\{q_{M_1}, q_{M_2}, \ldots, q_{M_\ell}\}, E_P)$ is the graph whose vertices are in a one-to-one correspondence to the modules in $P$. Two vertices $q_{M_i},q_{M_j}$ of $G_{/P}$ are adjacent if and only if the corresponding modules $M_i$ and $M_j$ are adjacent (with adjacency as above).
\end{definition}

If $P = \{M_1, M_2, \ldots, M_\ell\}$ is a modular partition of a graph $G$, then the quotient graph $G_{/P}$ is a compact representation of the edges with endpoint in different modules. Together with all subgraphs $G[M_i]$, with $i \in [\ell]$, we can reconstruct $G$. Each subgraph $G[M_i]$ is called a \emph{factor}. Instead of specifying the factors, one can recursively decompose them as well until one reaches trivial modules $\{v\}$. To make the decomposition unique, one considers modular partitions consisting of strong modules.
A module of a graph $G$ is called a \emph{strong} module, if it does not 
overlap with any other module of $G$. One can represent all strong modules of a graph $G$ by an inclusion tree $MD(G)$. 
Each strong module $M$ in $G$ corresponds to a vertex $v_M$ in $MD(G)$. 
A vertex $v_A$ is an an ancestor of $v_B$ in $MD(G)$ if and only if $B \subsetneq A$ for the corresponding strong modules $A$ and $B$ of $G$. Hence, the root node of $MD(G)$ corresponds always to the complete vertex set $V$ of $G$ and every leaf of $MD(G)$ corresponds a singleton set $\{v\}$ with $v \in V$.  Consider an internal node $v_M$ of $MD(G)$ with the set of children $\{v_{M_1}, \ldots, v_{M_\ell}\}$, i.e., $v_M$ corresponds to a strong module $M$ of $G$ and $P = \{M_1, \ldots , M_\ell\}$ is a modular partition of $G[M]$ into strong modules where $M_i$ is the corresponding module of $v_{M_i}$, with $i \in [\ell]$. There are three types of internal nodes in $MD(G)$. 
A node $v_M$ in $MD(G)$ is \emph{degenerate}, if for any non-empty subset of the children of $v_M$ in $MD(G)$, the union of the corresponding modules induces a (not necessarily strong) module. In this case the quotient graph $G[M]_{/P}$ is either a clique or an independent set. In the former case one calls $v_M$ a \emph{parallel} node, in the later a \emph{series} node. Another case are so called \emph{prime} nodes. Here, for no proper subset of the children of $v_M$, the union of the corresponding modules induces a module. In this case the quotient graph of $v_M$ is prime.
Gallai showed there are no further nodes in $MD(G)$.

\begin{theorem}[\cite{gallai1967transitiv}] \label{modular_decomposition_theorem}
 For any graph $G = (V, E)$ one of the three conditions is satisfied:
 \begin{itemize}
  \item $G$ is not connected,
  \item $\overline{G}$ is not connected,
  \item $G$ and $\overline{G}$ are connected and the quotient graph $G_{/P}$, where $P$ is the maximal modular partition of $G$, is a prime graph.
 \end{itemize}

\end{theorem}

Theorem~\ref{modular_decomposition_theorem} implies that $MD(G)$ is unique. The tree $MD(G)$ is called the \emph{modular decomposition tree} and the \emph{modular-width}, denoted by $\mw=\mw(G)$, is the minimum $k \geq 2$ such that any prime node in $MD(G)$ has at most $k$ children. Since every node in $MD(G)$ has at least two children and there are exactly $n$ leaves, $MD(G)$ has at most $2n -1$ nodes.
It is known that $MD(G)$ can be computed in time $\Oh(n + m)$ \cite{tedder2008simpler}. We refer to a survey of Habib and Paul~\cite{habib2010survey} for more information.

\section{Maximum Matching}\label{sec:maximummatching}

In the \prob{Maximum Matching} problem we are given a graph $G=(V,E)$ and need to find a maximum set $X\subseteq E$ of pairwise disjoint edges.
The size of a maximum matching of a graph $G$ is denoted by $\mu(G)$.
Edmond~\cite{edmonds1965paths} was the first to give a polynomial-time algorithm for this problem.
The fastest known algorithm, due to Micali and Vazirani~\cite{MicaliV80}, runs in time $\Oh(m \sqrt{n})$.
A \emph{$b$-matching} is a generalization of a matching that specifies for each vertex a \emph{degree bound} of how many edges in the matching may be incident with that vertex. Formally, degree bounds are given by a function $b\colon V \rightarrow \mathbb{N}$, and a $b$-matching is a function $x \colon E \rightarrow \mathbb{N}$ that fulfills for every vertex $v \in V$ the constraint that $\sum_{e \in \delta(v)}x(e) \leq b(v)$. Gabow~\cite{Gabow16a} showed how to find a $b$-matching that maximizes $\sum_{e \in E} x(e)$ in time $\Oh((n \log n)\cdot(m+n \log n))$.  

Recently, Coudert et al.~\cite{CoudertDP18} gave an $\Oh(\mw^4  n + m)$ time algorithm for \prob{Maximum Matching}, where $\mw$ denotes the modular-width of the input graph. In the following we will improve this result by providing an algorithm for \prob{Maximum Matching} that runs in time $\Oh({\mw^2 \log\mw} \cdot n + m)$.
The main idea of our algorithm is to compress the computation of a matching in $G$ to a computation of a $b$-matching, instead of using the structure of modular decompositions to speed up the search for augmenting paths (like in~\cite{CoudertDP18}).

\begin{theorem}\label{Thm:matching}
 For every graph $G = (V, E)$ with modular-width $\mw$, \prob{Maximum Matching} can be solved in time $\Oh({\mw^2 \log \mw} \cdot n  + m)$.
\end{theorem}

\subparagraph{Algorithm.}

First, we compute the modular decomposition tree $MD(G)$. We will traverse the decomposition tree in a bottom-up manner. For each $v_M$ in $MD(G)$, with $M$ denoting the corresponding module of $G$, we will compute a maximum matching in $G[M]$. Note that for the root module $v_M$ of $MD(G)$ it holds that $G[M] = G$. 
For any leaf module $v_M$ of $MD(G)$, we have $\mu(G[M]) = 0$, since $G[M]$ is a graph consisting of a single vertex. 
Let $v_M$ be a non-leaf vertex in $MD(G)$ with the set of children $\{v_{M_1}, \ldots, v_{M_\ell}\}$. This means that $\{M_1, \ldots, M_\ell\}$ is a modular partition of $G[M]$, where $M_i \subseteq M$  corresponds to the vertex $v_{M_i}$ in $MD(G)$ for $i \in [\ell]$. 
In the following, we can always assume that we have already computed $\mu(G[M_i])$ for $i \in [\ell]$. The next lemma shows that the concrete structure inside a module is irrelevant for the maximum matching size of the whole graph, i.e., only the number of vertices and the maximum matching size is important.  
The lemma is a more general version of \cite[Lemma 5.1]{CoudertDP18}, but can be proven in a similar way.

\begin{lemma} \label{lem:moduleReplace}
 Let $M$ be a module of $G = (V, E)$ and let $G[M] = (M, E_M)$. Let $A \subseteq\binom{M}{2}$ be any set of edges on the vertices of $M$ such that $\mu((M, A)) =\mu((M, E_M))$. Then, the size of a maximum matching of $G' = (V, (E \setminus E_M) \cup A)$ is equal to the size of a maximum matching of $G$. 
\end{lemma}

\begin{proof}
We first show that $\mu(G')\geq \mu(G)$.
Let us consider a maximum matching $F \subseteq E$ in $G = (V, E)$. 
To get a maximum matching in $G'$ we replace all edges in $F$ that are incident with $M$:
First, replace all edges in $F\cap E(G[M])$ by an arbitrary matching $A'\subseteq A$ of the same size; such a matching must exist because $F\cap E(G[M])$ is not larger than a maximum matching in $G[M]$ and $\mu((M, A)) =\mu((M, E_M))$. 
Second, we replace all edges in $F$ that have exactly one endpoint in $M$ as follows: Let $X \subseteq M$ be the set of vertices in $M$ that are endpoints of an edge in $F$ whose other endpoint is not in $M$. By assumption, $\vert M \setminus V(A')\vert \geq \vert X \vert$ and since all vertices in $V \setminus M$ that are connected to a vertex in $X$ in $G$ are also connected to all vertices in $M \setminus V(A')$ in $G'$, we can replace all edges of $F$ that have exactly one endpoint in $M$.
 Thus, $\mu(G')\geq\mu(G)$, i.e., replacing the edges in a module by an arbitrary set of edges with same maximum matching size does not decrease the size of the maximum matching for the whole graph. Applying this argument for $A':=E_M$ to swap back to the original edge set yields, $\mu(G)\geq\mu(G')$ and completes the proof.
\end{proof}

We now describe how to compute $\mu(G[M])$ for a node $v_M$ in $MD(G)$. Let $\{v_{M_1}, \ldots, v_{M_\ell}\}$ be the set of children of $v_M$ in $MD(G)$, meaning that $P = \{M_1, \ldots, M_\ell\}$ is a modular partition of $G[M]$. We can assume that we have already computed $\mu(G[M_i])$ for $i \in [\ell]$. Let $G[M]_{/P}$ be the quotient graph of $G[M]$. If $v_M$ is a parallel node then $G[M]_{/P}$ is edgeless, i.e., $G[M]$ is the disjoint union of all $G[M_i]$. In this case a maximum matching for $G[M]$ simply consists of the union of maximum matchings for each $G[M_i]$ and we set $\mu(G(M)) = \sum_{i \in [\ell]} \mu(G[M_i])$. Next, suppose that $v_M$ is a prime node. We will reduce the problem of computing a maximum matching in $G[M]$ to computing a maximum $b$-matching in an auxiliary graph closely related to the quotient graph of $v_M$ that we will define next.

\begin{definition} \label{def:instance_b_matching}
 Let $G = (V, E)$ be a graph and $P = \{M_1, \ldots, M_\ell\}$ be a modular partition of $G$. Let $n_i$ denote the number of vertices in $G[M_i]$ and $f_i$ the size of a maximum matching in $G[M_i]$. We define an auxiliary graph $G^* = (V^*, E^*)$ together with degree bounds $b^* \colon V^* \rightarrow \mathbb{N}$ as an instance $(G^*, b^*)$ for the maximum $b$-matching problem as follows:
 \begin{itemize}
  \item For every module $M_i \in P$, with $i \in [\ell]$, we add three vertices $v_i^1, v_i^2, v_i^3$ to $V^*$ and set $b^*(v_i^1) = b^*(v_i^2) = f_i$ and $b^*(v_i^3) = n_i - 2f_i$.
  \item We add the edge $\{v_i^1, v_i^2\}$ to $E^*$ for $i \in [\ell]$.
  \item For each edge between vertices $q_i$ and $q_j$ in $G_{/P}$ that corresponds to modules $M_i$ and $M_j$, we add the nine edges $\{v_i^c, v_j^d\}$ with $c,d \in \{1,2,3\}$ to $E^*$.
 \end{itemize} 
\end{definition}

\begin{lemma}\label{lem:compressiontobmatching}
 Let $G = (V, E)$ be a graph and $P = \{M_1, \ldots, M_\ell\}$ be a modular partition of $G$. Let $(G^*, b^*)$ be the instance of a maximum $b$-matching problem as defined in Definition~\ref{def:instance_b_matching}. Then the size of maximum matchings in $G$ is equal to the size of a maximum $b$-matching of $(G^*, b^*)$. 
\end{lemma}

\begin{proof}
 Consider a graph $G = (V, E)$ with a modular partition $P = \{M_1, \ldots, M_\ell\}$. For $M_i \in P$ let $n_i = \vert V(G[M_i]) \vert$ and let $f_i = \mu(G[M_i])$. Due to Lemma~\ref{lem:moduleReplace}, we can replace each $G[M_i]$, for $i \in [\ell]$, by a graph consisting of a complete bipartite graph $K_{f_i,f_i}$ together with $n_i - 2f_i$ single vertices without changing the size of a maximum matching. We do this for every module $M_i \in P$ and denote the resulting graph by $\overline{G}$. Note, that $\mu(G) = \mu(\overline{G})$. Now, each replacement of $G[M_i]$ can be partitioned into three modules, giving us a modular partition $P'$ of $\overline{G}$  of size $3\ell$, and for every module $M \in P'$ the factor graph $G[M]$ is an independent set. The quotient graph $\overline{G}_{/P'}$ is exactly the auxiliary graph $G^*$ of $G$ and the degree bound of a vertex $v$ in $G^*$ is equal to the number of vertices in the corresponding module. Since solving a $b$-matching in $(G^*, b^*)$ directly corresponds to solving maximum matching in $\overline{G}$, this completes the proof.
\end{proof}

Finally, suppose that $v_M$ is a series node. Instead of computing $\mu(G[M])$ directly, we will modify the decomposition tree $MD(G)$ (cf.~\cite{CoudertDP18}). Let $\{v_{M_1}, \ldots,v_{M_\ell}\}$ be the children of $v_M$ in $MD(G)$. We will iteratively compute a maximum matching for $G_i = G[\cup_{1 \leq j \leq i} M_j]$ by using a modular partition of $G_i$ consisting of the two modules $\cup_{1 \leq j < i} M_j$ and $M_i$, for $i \in [\ell]$. This means that we replace a series node with $\ell$ children by $\ell-1$ series nodes with only two children. We will treat the newly inserted nodes as prime nodes (with a quotient graph isomorphic to $K_2$).  After replacing the series nodes of the modular decomposition tree $MD(G)$, every node still has at least two children; hence, we still have a most $2n-1$ nodes in $MD(G)$.

\subparagraph{Running Time.}

Consider a graph $G = (V, E)$ with modular-width $\mw$. Computing the modular decomposition tree $MD(G)$ takes time $\Oh(n + m)$. Since there are at most $2n -1$ nodes in $MD(G)$ the total computation for all parallel nodes together takes time $\Oh(n)$. As described above, we modify the decomposition tree such that every series node of $MD(G)$ with $\ell \geq 3$ children is replaced by $\ell - 1$ `pseudo-prime` nodes with exactly two children. This replacement can be done in time $O(n)$.
Now, every node $v_M \in MD(G)$ that is not a parallel node has a set of children $\{v_{M_1}, \ldots, v_{M_\ell}\}$ with $\ell \leq \mw$. This means that $P = \{M_1, \ldots, M_\ell\}$ is a modular partition of $G[M]$ and the quotient graph $G[M]_{/P}$ consists of $\ell \leq \mw$ vertices. Since we have already computed $\mu(G[M_i])$ for all $i \in [\ell]$, we can construct the auxiliary graph $G^*$ of $G[M]$ as defined in Definition~\ref{def:instance_b_matching} in time $\Oh(V(G^*)+ E(G^*)) = \Oh(\ell^2)$. Recall, that $\vert V(G^*) \vert = 3\ell$. Thus, we can compute a maximum $b$-matching of $G^*$ subject to $b$ in time $\Oh(\ell^3 \log \ell)$ using the algorithm due to Gabow~\cite{Gabow16a}. We have to do this for every prime and series node, but a slightly more careful summation of running times over all nodes gives an improvement over the obvious upper bound of $\Oh({\mw^3\log \mw} \cdot n +m)$:
Let $t$ be the number of nodes in $MD(G)$ and for a node $v_{M_i}$ in $MD(G)$ let $\ell_i$ denote the number of children, i.e. the number of vertices of the quotient graph of $G[M_i]$. Then, neglecting constant factors and assuming that $MD(G)$ is already computed, we can solve \prob{Maximum Matching}, in time:
\begin{align*}
 \sum_{i = 1}^{t} \ell_i^3 \log \ell_i \leq \left(\sum_{i = 1}^{t} \ell_i\right) \cdot \max_{i \in [t]} \{ \ell_i^2 \log \ell_i \} \leq 2n \cdot \max_{i \in [t]} \{ \ell_i^2 \log \ell_i \} \leq 2n \cdot (\mw^2 \log \mw)
\end{align*}

The second inequality holds, since $\sum_{i = 1}^{t} \ell_i$ counts each node in $MD(G)$ once, except for the root. Since constant factors propagate through the inequality, the total running time of the algorithm is $\Oh({\mw^2 \log \mw} \cdot n + m)$, which proves Theorem~\ref{Thm:matching}.

\subparagraph{Generalization to b-matching}
We can easily generalize this result to the more general maximum $b$-matching problem.
\begin{theorem}
 For every graph $G = (V,E)$ with modular-width $\mw$, \prob{Maximum $b$-matching} can be solved in time $\Oh({\mw^2 \log \mw} \cdot n + m)$.
\end{theorem}
Again, the concrete structure inside a module will not be important. The only important information is the size of a maximum $b$-matching and the sum of all $b$-values in a module. We naturally extend Definition~\ref{def:instance_b_matching} to $b$-matchings:

\begin{definition} \label{def:instance_b_matching_2}
 Let $G = (V, E)$ be a graph with $b \colon V \rightarrow \mathbb{N}$ and let $P = \{M_1, \ldots, M_\ell\}$ be a modular partition of $G$. Let $n_i = \sum_{v \in M_i} b(v)$ and $f_i$ be the size of a maximum $b$-matching in $G[M_i]$ for $i \in [\ell]$. We define the auxiliary graph $G^* = (V^*, E^*)$ together with degree bounds $b^*\colon V \rightarrow \mathbb{N}$ in the same way as done in Definition~\ref{def:instance_b_matching}.
\end{definition}

\begin{lemma}
 Let $G = (V, E)$ be a graph and $P = \{M_1, \ldots, M_\ell\}$ be a modular partition of $G$. Let $(G^*, b^*)$ be the instance of a maximum $b$-matching problem as defined in Definition~\ref{def:instance_b_matching_2}. Then the size of a maximum $b$-matching in $(G,b)$ is equal to the size of a maximum $b$-matching of $(G^*, b^*)$. 
\end{lemma}

\begin{proof}
 Consider a graph $G = (V, E)$ with a modular partition $P = \{M_1, \ldots, M_\ell\}$. For $M_i \in P$ let $n_i = \sum_{v \in M_i} b(v)$ and let $f_i$ be the size of a maximum $b$-matching in $M_i$.
 Note, that one can solve $b$-matching by replacing every vertex $v$ by $b(v)$ copies that are connected in the same way as $v$. After considering this replacement and 
 due to Lemma~\ref{lem:moduleReplace}, we can replace $G[M_i]$, for $i \in [\ell]$, by a graph consisting of a complete bipartite graph $K_{f_i,f_i}$ together with $n_i - 2f_i$ single vertices without changing the size of a maximum matching. We do this for every module $M_i$ and denote the resulting graph by $\overline{G}$. As in the proof of Lemma~\ref{lem:compressiontobmatching}, we can subdivide every module in three parts. This yields to the instance $(G^*, b^*)$ as defined in Definition~\ref{def:instance_b_matching_2}. Again, solving a maximum $b$-matching of $(G^*,b^*)$ directly corresponds to solving a maximum $b$-matching in $\overline{G}$, which completes the proof.
\end{proof}
The running time can be bounded in the same way as before. However, to see that this algorithm is also adaptive for sparse graphs (at least for large $b$-values), we can modify the computation of the running time: 
Let $t$ be the number of nodes in $MD(G)$. For a node $v_{M_i}$ in $MD(G)$ let $n_i$ denote the number of vertices and $m_i$ denote the number of edges of the quotient graph of $G[M_i]$.
Thus, we can compute a maximum $b$-matching of $G^*$ subject to $b^*$ in time $\Oh((n_i \log n_i) \cdot (m_i + n_i \log n_i))$ using the algorithm due to Gabow~\cite{Gabow16a}.
Then, neglecting constant factors and assuming that $MD(G)$ is already computed, we can solve \prob{Maximum $b$-Matching} in time:
\begin{align*}
 \sum_{i = 1}^{t} (n_i \log n_i) \cdot (m_i + n_i \log n_i) &= \sum_{i = 1}^{t} m_i n_i \log n_i + \sum_{i = 1}^{t} n_i^2 \log^2 n_i \\
 &\leq \left(\sum_{i = 1}^{t} m_i \right)  \max_{i \in [t]} \{ n_i \log n_i \} + \left(\sum_{i = 1}^{t} n_i\right)  \max_{i \in [t]} \{n_i \log^2 n_i\} \\
 &\leq m \cdot {\mw \log \mw}  + 2n \cdot (\mw \log^2 \mw)
\end{align*}
 Since constant factors propagate through the inequality, the total running time of the algorithm is $\Oh((m + n \log \mw) \cdot (\mw \log \mw))$. Therefore, even for $\mw = \Theta(n)$ our algorithm is not worse than the (currently) best unparameterized algorithm, assuming $b(V) \geq n \log n$, where $b(V) = \sum_{v \in V}b(v)$.
 
\section{Triangle Counting}\label{sec:trianglecounting}

In this section we consider the $\prob{Triangle Counting}$ problem, in which one is interested in the number of triangles in the input graph.
\begin{theorem}\label{thm:trianglecounting}
For every graph $G = (V, E)$ with modular-width $\mw$, \prob{Triangle Counting} can be solved in time $\Oh(n \cdot {\mw^{\omega - 1}} + m )$.
\end{theorem}    
\subparagraph{Algorithm}
First, we compute the modular decomposition tree $MD(G)$. We will process $MD(G)$ in a bottom-up manner. For each $v_{M}$ in  $MD(G)$, with corresponding module $M$ in $G$, we will compute the following three values: the number of vertices $n_M =\vert V(G[M]) \vert$, the number of edges $m_M = \vert E(G[M]) \vert$, and the number of triangles $t_M$ in $G[M]$. For any leaf node $v_M$ in $MD(G)$ we have $n_M = 1$ and $m_M = t_M = 0$, because $G[M]$ consists of a single vertex. Let $v_M$ be a non-leaf node in $MD(G)$ with children $\{v_{M_1}, \ldots, v_{M_\ell}\}$. Since we process $MD(G)$ in a bottom-up manner, the values for $G[M_i]$ are already computed for $i \in [\ell]$.
If $v_M$ is a parallel node, the values simply add up, i.e. $n_M = \sum_{i = 1}^\ell n_{M_i}$, $m_M = \sum_{i = 1}^\ell m_{M_i}$, and $t_M = \sum_{i = 1}^\ell t_{M_i}$. If $v_M$ is a series node, we will use the same approach as in Section~\ref{sec:maximummatching} and replace $v_M$ by $\ell - 1$ series nodes with only two children each. Afterwards, we compute the values for a series node $v_M$ with children $v_{M_1}$ and $v_{M_2}$ as follows:
\begin{align*}
n_M &= n_{M_1}+n_{M_2}  \\
m_M &= m_{M_1} + m_{M_2} + n_{M_1}n_{M_2}  \\
t_M &= t_{M_1} + t_{M_2} + m_{M_1}n_{M_2} + m_{M_2}n_{M_1}
\end{align*}
Finally, let $v_M$ be a prime node in $MD(G)$ and let $\{v_{M_1},\ldots,v_{M_\ell}\}$ be the children of $v_M$ in $MD(G)$. This means that $P = \{M_1, \ldots, M_\ell\}$ is a modular partition of $G[M]$. Again, $n_M = \sum_{i = 1}^\ell n_{M_i}$ and we can compute $m_M$ by traversing all edges in the quotient graph $G[M]_{/P}$, i.e., $m_M = \sum_{i = 1}^\ell m_{M_i} + \sum_{\{q_i, q_j\} \in E(G[M]_{/P})} n_{M_i}n_{M_j}$. For computing $t_M$ we count triangles in $G[M]$ of three types: Triangles using vertices in exactly one module, in two (adjacent) modules, or in three modules of $P$. We call a triangle with vertices in three different modules a \emph{separated} triangle.  To compute the number of separated triangles, we use the following lemma:

\begin{lemma}\label{lem:separatedtriangles}
Let $G = (V, E)$ be a graph with a modular partition $P = \{M_1, \ldots, M_\ell\}$ and quotient graph $G_{/P}$. Let $n_{M_i} := \vert M_i \vert$ and consider the weight function $w \colon E(G_{/P}) \rightarrow \mathbb{R}^+$ with $w(\{q_i, q_j\}) = \sqrt{n_{M_i} n_{M_j}}$. Let $A$ be the weighted adjacency matrix of $G_{/P}$ with respect to $w$. Then, the number of separated triangles in $G_{/P}$ is:
\begin{align*}
  \sum_{i,j=1}^\ell \frac{1}{3}(A^2 \circ A)_{i,j} ,
\end{align*}
where $A \circ B$ denotes the Hadamard product of the matrices $A$ and $B$, i.e., $(A \circ B)_{i,j} = A_{i,j} B_{i,j}$.
\end{lemma}

\begin{proof}
To count all separated triangles in $G$ we need to sum up the values $n_{M_i} n_{M_j} n_{M_k}$ for each triangle $(q_i, q_j, q_k)$ in $G_{/P}$.
We show, that $(A^2 \circ A)_{i,j}$ exactly corresponds to the number of separated triangles in $G$ with two vertices in $M_i$ and $M_j$; here, a \emph{wedge} is a path on three vertices (and a wedge $(q_i, q_k, q_j)$ requires the presence of edges $\{q_i,q_k\}$ and $\{q_k,q_j\}$):
\begin{align*}
\left(A^2\right)_{i,j} &= \sum_{k=1}^\ell A_{i,k} A_{k,j}  \\
			&= \sum_{\substack{  k: (q_i, q_k, q_j)\\ \text{is a wedge in } G_{/P}}} \sqrt{n_{M_i} n_{M_k}} \sqrt{n_{M_k} n_{M_j}}  \\		    
		    &=  \sqrt{n_{M_i} n_{M_j}} \sum_{\substack{  k: (q_i, q_k, q_j)\\ \text{is a wedge in } G_{/P}}} n_{M_k} \\
\Rightarrow \left(A^2 \circ A\right)_{i,j} &= \sum_{\substack{k: (q_i, q_k, q_j)\\ \text{is a triangle in } G_{/P}}} n_{M_i} n_{M_j} n_{M_k} 
\end{align*}
Since every triangle is counted three times (once for each edge) the lemma follows.
\end{proof}
Using Lemma~\ref{lem:separatedtriangles}, we can compute $t_M$ by
\begin{align*}
t_M = \sum_{i = 1}^\ell t_{M_i} + \sum_{\{q_i, q_j\} \in E(G_{/P})} \left( m_{M_i}n_{M_j}+n_{M_i}m_{M_j} \right) +\sum_{i,j=1}^\ell \frac{1}{3}\left(A^2 \circ A\right)_{i,j},
\end{align*}
where the three terms refer to triangles with vertices from only one module, triangles using vertices of two adjacent modules, and separated triangles with vertices in three different (pairwise adjacent) modules.

\subparagraph{Running Time.}
Computing the modular decomposition tree $MD(G)$ takes time $\Oh(n + m)$. Consider a node $v_M$ in $MD(G)$ with children $\{v_{M_1}, \ldots, v_{M_\ell}\}$. If $v_M$ is a parallel or a series node then we can compute the values $n_M$, $m_M$, and $t_M$ for $G[M]$ in time $\Oh(\ell)$. Thus, since the number of nodes in $MD(G)$ is at most $2n -1$, the total running time for all parallel and series nodes is $O(n)$. Assume that $v_M$ is a prime node. Recall, that $P = \{M_1, \ldots, M_\ell\}$ is a modular partition of $G[M]$. Computing $n_M$ takes time $\Oh(\ell)$ and computing $m_M$ takes time $\Oh(\vert E(G[M]_{/P}) \vert) = \Oh(\ell^2)$. The running time for computing $t_M$ is dominated by the computation of $A^2$, which takes time $\Oh(\ell^\omega)$. Note, that $2 \leq \ell \leq \mw$.
By a similar careful summation as done in Section~\ref{sec:maximummatching} we can improve the obvious upper bound of $\Oh(n \cdot {\mw^\omega} + m)$:
Let $p$ be the number of nodes in $MD(G)$ and for $v_{M_i}$ in $MD(G)$ let $\ell_i$ be the number of children, i.e., the number of vertices of the quotient graph of $G[M_i]$. Neglecting constant factors and assuming that $MD(G)$ is already computed, the running time is:
\begin{align*}
 \sum_{i=1}^p \ell_i^\omega \leq \left(\sum_{i=1}^p \ell_i\right)  \max_{i \in [p]} \ell_i^{\omega-1} \leq  2n \cdot \mw^{\omega-1}
\end{align*}
Again, since constant factors propagate through the inequalities, the total running time of the algorithm is $\Oh(n \cdot {\mw^{\omega-1}} + m)$, which proves Theorem~\ref{thm:trianglecounting}. Note, that this algorithm is adaptive for dense graphs, meaning that even for $\mw = \Theta(n)$ our algorithm is not worse than $\Oh(n^\omega)$.

\section{Edge-Disjoint Paths}\label{sec:edgedisjointpaths}

In this section, we address the problems of finding the maximum number of edge-disjoint $s$-$t$ paths (equivalently, finding an unweighted minimum $s$-$t$ cut) in a given graph $G$. We denote the size of a maximum number of edge-disjoint $s$-$t$ paths in a graph $G$ by $\lambda_G(s,t)$. Later, we focus on finding a global unweighted minimum cut. The weighted variants of these problems, in particular \prob{Maximum $s$-$t$ Flow}, are unlikely to admit faster algorithms when the modular-width is low (see conclusion).

\subsection{Maximum Edge-Disjoint s-t Paths}\label{sec:stedgecut}

\begin{theorem}\label{thm:edgedisjointpaths}
 For every graph $G = (V, E)$ with modular-width  $\mw$, \prob{Edge-Disjoint $s$-$t$ Paths} can be solved in time $\Oh({\mw^3} + n +m)$.
\end{theorem}

\subparagraph{Algorithm.}
Let $G = (V, E)$ be a graph with $s, t \in V$. We assume that $G$ is connected (otherwise consider the connected component with $s$ and $t$ as the new input graph or conclude that $\lambda_G(s,t) = 0$ if $s$ and $t$ are in different connected components of $G$).
First, we compute the modular decomposition tree $MD(G)$. Instead of traversing the decomposition tree in a bottom-up manner, we will only consider one modular partition of $G$. 
\begin{lemma}\label{lem:edgedisjointsamemodule}
 Let $G = (V,E)$ be a graph, let $s, t \in V$, and let $P$ be  a modular partition of $G$. If there exists a module $M \in P$ with $s, t \in M$ and a module $N \in P$ that is adjacent to $M$, then 
 $\lambda_G(s,t) = \min \{\deg_G(s), \deg_G(t)\}$.
\end{lemma}
\begin{proof}
 Obviously, it holds that  $\lambda_G(s,t) \leq \min \{\deg_G(s), \deg_G(t)\}$. For the converse direction assume, w.l.o.g.~, that $\deg_G(s) \leq \deg_G(t)$. For every vertex $v \in N_G(s) \setminus M$, we consider the path $(s,v,t)$. We recall that $v$ is also a direct neighbor of $t$ and that all these paths are clearly edge-disjoint. 
 Since we assume that $\deg_G(s) \leq \deg_G(t)$, it also holds that $\deg_{G[M]}(s) \leq \deg_{G[M]}(t)$. Hence, $\vert N_{G[M]}(s)\vert \leq \vert N_{G[M]}(t)\vert$ and we can assign for every vertex $v \in N_{G[M]}(s)$ a private vertex $v' \in N_{G[M]}(t)$. Let $w \in N$ be an arbitrary vertex in the neighboring module $N$. For all $v \in N_{G[M]}(s)$ we either choose the path $(s, v, t)$, if $v = v'$, or the path $(s,v,w,v',t)$, if $v \neq v'$. 
 For $v=t\in N_M(s)$, if it exists, we use the path $(s,t)$. Overall this results in $\deg_G(s)$ many edge-disjoint $s$-$t$ paths. This implies that $\deg_G(s) = \min \{\deg_G(s), \deg_G(t)\} \leq \lambda_G(s,t)$.
\end{proof}
\begin{corollary}\label{cor:edgedisjointseriesnode}
 Let $G = (V, E)$ be a graph, let $s, t \in V$, and let $P$ be a modular partition of $G$ such that $G_{/P}$ is a complete graph. Then $\lambda_G(s,t) = \min \{\deg_G(s), \deg_G(t)\}$.
\end{corollary}
Corollary~\ref{cor:edgedisjointseriesnode} follows from Lemma~\ref{lem:edgedisjointsamemodule} since there exists always a modular partition $P$ such that $s$ and $t$ are in the same module, if $G_{/P}$ is a complete graph.

Consider the root vertex $v_M$ of $MD(G)$ and let $\{v_{M_1}, \ldots, v_{M_\ell}\}$ be the children of $v_M$. This means, that $P = \{M_1, \ldots, M_\ell\}$ is a modular partition of $G[M] = G$.
Since we assume that $G$ is connected, $v_M$ cannot be a parallel node. If $v_M$ is a series node, we can conclude that $\lambda_G(s,t) = \min \{\deg_G(s), \deg_G(t)\}$ by Corollary~\ref{cor:edgedisjointseriesnode}. Let $v_M$ be a prime node. If $s$ and $t$ belong to the same module, we again conclude that $\lambda_G(s,t) = \min \{\deg_G(s), \deg_G(t)\}$ due to Lemma~\ref{lem:edgedisjointsamemodule}, since every quotient graph of a prime node is connected. It remains to solve the case that $v_M$ is a prime node but $s$ and $t$ do not belong to the same module. 
We will reduce this case to a single computation of a maximum edge-capacitated flow. 
Recall, that we denote the set of vertices of the quotient graph $G_{/P}$ by $\{q_1, \ldots, q_\ell\}$ and each vertex $q_i$ corresponds to the module $M_i \in P$, for $i \in [\ell]$.

\begin{definition}\label{def:edgedisjointflownetwork}
 Let $G = (V, E)$ be a graph, let $s, t \in V$, and let $P = \{M_1, ..., M_\ell\}$ be a modular partition of $G$ into $\ell \geq 2$ modules. Let $s \in M_1$, let $t \in M_\ell$, and  let $G_{/P}$ be the quotient graph with vertex set $\{q_1, \ldots, q_\ell\}$.
 We define a flow network $N = (G', q_0, q_{\ell+1}, c)$ as follows:
 \begin{itemize}
  \item The graph $G'$ is initiated as being equal to $G_{/P}$.
  \item We add vertices $q_0$ and $q_{\ell+1}$ to $V(G_{/P})$, each with the same neighbors as $q_1$ resp. $q_\ell$.
  \item We add the edges $\{q_0, q_1\}$ and $\{q_{\ell + 1}, q_\ell\}$.
  \item The $\ell + 2$ vertices of $G'$ correspond to the sets of vertices in the partition \\$P' = \{M'_0,M'_1,M'_2, \ldots, M'_\ell, M'_{\ell + 1}\}$ with $M'_0 = \{s\}$, $M'_1 = M_1 \setminus \{s\}, M'_\ell = M_\ell \setminus \{t\}$, $M'_{\ell+1} = \{t\}$, and $M'_i = M_i$ for $i \in \{2,3,\ldots, \ell-1\}$. 
  \item The capacities on the edges of $G'$ represent the number of edges between the corresponding vertex sets in $G$, i.e. $c(q_i, q_j) = \vert M'_i \vert \vert M'_j \vert$ for $\{q_i, q_j\} \in E(G')\setminus \{\{q_0, q_1\}, \{q_\ell, q_{\ell+1}\}\}$  resp. $c(q_0, q_1) = deg_{G[M_1]}(s)$ and $c(q_\ell, q_{\ell+1}) = deg_{G[M_1]}(s)$.
 \end{itemize}
\end{definition}
\begin{lemma}\label{lem:primeflow}
 Let $G = (V, E)$ be a graph, let $s, t \in V$, and let $P = \{M_1, ..., M_\ell\}$ be a modular partition of $G$. Let $s \in M_1$, $t \in M_\ell$ and $\ell \geq 2$. Let $N = (G', q_0, q_{\ell+1}, c)$ be the flow network as defined in Definition~\ref{def:edgedisjointflownetwork}. Then, the maximum flow in $N$ is equal to $\lambda_G(s,t)$.
\end{lemma}

The graph $G'$ together with the capacities $c$ is a compact representation of $G$, but without the edges inside a module (except for incident edges to $s$ or $t$). In order to prove Lemma~\ref{lem:primeflow}, we first observe that those edges inside modules are not helpful to get edge-disjoint paths in $G$. To see this, we consider the following assignment problem.

\begin{lemma}\label{lem:assignmentproblem}
 Let $A = \{a_1, \ldots, a_\ell\}$, let $B = \{b_1, \ldots,b_r\}$, and let $X = \{x_1, \ldots, x_k\}$ be sets of vertices and $G$ be a graph with vertex set $A \cup B \cup X$. Let $f\colon A \cup B \rightarrow \mathbb{N}$ be a function that denotes the demand of every vertex in $A \cup B$, with the constraints $f(a_i) \leq k = \vert X \vert$ and $f(b_j)\leq k$ for all $i \in [\ell]$ and $j \in [r]$, and $\sum_{i = 1}^\ell f(a_i) = \sum_{j=1}^r f(b_j)$. Then there is a set of directed arcs $E \subseteq (A \times X) \cup (X \times B)$ in $G$ such that $\delta_G^+(a_i) = f(a_i)$, $\delta_G^-(b_j) = f(b_j)$ and $\delta_G^+(x_d) = \delta_G^-(x_d)$ for all $i \in [\ell]$, $j \in [r]$ and $d \in [k]$.
\end{lemma}

\begin{proof}
 We can solve this assignment problem with a flow computation. Therefore, we construct a directed graph $G$ as follows: The vertex set consists of $A$, $B$, $X$ and two vertices $s$ and $t$. We add edges $(s,a_i)$ of capacity $f(a_i)$ for each $i \in [\ell]$ and denote these edges with $S$. In almost the same manner we add edges $(b_j,t)$ of capacity $f(b_j)$ for each $j \in [r]$ and denote these edges with $T$. At last, we add all edges $A \times X$ and $X \times B$ to the graph, each with capacity one. Denote the resulting network by $N = (G, s, t, c)$. 
 To prove the lemma we only have to show that the maximum flow in $N$ is equal $\sum_{i = 1}^\ell f(a_i)= c(S) = c(T)$.
 However, one can observe that the minimum weighted $s$-$t$ cut in $N$ is equal to $c(S) = c(T)$: 
 Let $C \subseteq E(G)$ be an arbitrary minimum $s$-$t$ cut in $G$. 
 If $S \subseteq C$ then it holds that $c(S) \leq c(C)$ and since $S$ is an $s$-$t$ cut there is indeed equality. The same applies if $T \subseteq C$. Thus, assume that $S \setminus C \neq \emptyset$ and $T \setminus C \neq \emptyset$. Let w.l.o.g.~$\vert S \setminus C \vert \leq \vert T \setminus C \vert$. Let $D = C \cap (S \cup T)$ and let $A' \subseteq A$, resp. $B' \subseteq B$, be the set of vertices of $A$, resp.$B$, that are not incident to an edge in $D$. Since $\vert S \setminus C \vert \leq \vert T \setminus C \vert$ it holds that $\vert A' \vert \leq \vert B' \vert$. It is easy to see that there are $k \cdot \vert A' \vert$ edge disjoint paths between $A'$ and $B'$. Hence, to augment $S \cap C$ to an $s$-$t$ cut without taking edges in $S \cup T$ one needs to take at least $k \cdot \vert A' \vert = k \cdot \vert S \setminus C \vert$ edges. Therefore, $c(C) \geq c(S \cap C) + k \cdot \vert S \setminus C \vert \leq c(S \cap C) + c(S \setminus C) = c(S)$. Again, since $S$ is an $s$-$t$ cut in $N$, there is indeed equality.  
\end{proof}

\begin{corollary}\label{cor:noedgesusedinmodules}
 Let $G = (V, E)$ be a graph, let $s, t \in V$, and let $P = \{M_1, ..., M_\ell\}$ be a modular partition of $G$ into $\ell \geq 2$ modules. Let $s \in M_1$, let $t \in M_\ell$. Let $P' = \{M'_0,M'_1,M'_2, \ldots, M'_\ell, M'_{\ell + 1}\}$ be the partition of $V(G)$ as defined in Definition~\ref{def:edgedisjointflownetwork}. Then, every maximum set of edge-disjoint $s$-$t$ paths can be modified to a set of edge-disjoint $s$-$t$ paths of the same size and no path uses edges inside a vertex set $M \in P'$.  
\end{corollary}

\begin{proof}
 Assume for contradiction, there is a path in a maximum set of edge-disjoint $s$-$t$ paths that uses an edge inside a vertex set $M \in P'$. Note that $M'_0 = \{s\}$ and $M_{\ell+1} = \{t\}$, implying $M \neq M'_0$ and $M \neq M'_{\ell+1}$. Thus, every path traversing nodes in $M$ visits a vertex before and after $M$. Orient every path to a directed path from $s$ to $t$ (since the paths are edge-disjoint, this is possible). Denote the set of those directed edges by $D$. We can apply Lemma~\ref{lem:assignmentproblem} to rearrange the paths, such that no edge inside $M$ is used, by setting $X = M$, $A = \{v \in V \mid (v,m) \in D \wedge m \in M\}$ and $B = \{v \in V \mid (m,v) \in D \wedge m \in M\}$. Additionally, we set the demand $f(a) = \vert \{ m \in M \mid (a,m) \in D \}\vert$ for $a \in A$ and 
 $f(b) = \vert \{ m \in M \mid (m,b) \in D \}\vert$ for $b \in B$.
 By applying this for every vertex set $M \in P$ we have proven the claim. 
\end{proof}

\begin{proof}[Proof of Lemma~\ref{lem:primeflow}]
Let $MF(N)$ denote the maximum flow value in $N = (G', q_0, q_{\ell+1}, c)$. Let $P' = \{M'_0, M'_1, M'_2, \ldots, M'_{\ell}, M'_{\ell+1}\}$ be the partition of $V$ corresponding to the vertices in $G'$.
Any flow in $N$ corresponds to edge-disjoint paths in $G$ not using edges inside a set of $P'$, yielding $MF(N) \leq \lambda_G(s,t)$. Conversely, by Corollary~\ref{cor:noedgesusedinmodules} we can modify any maximum set of edge disjoint $s$-$t$ paths to a set of edge-disjoint $s$-$t$ paths that does not use edges inside a vertex set of $P'$. Again, any such set of edge-disjoint path corresponds to a flow in $N$, proving $MF(N) \geq \lambda_G(s,t)$.
\end{proof}

\subparagraph{Running Time.}
Consider a connected graph $G = (V, E)$ with modular-width $\mw$ and let $s, t \in V$. Computing the modular decomposition tree takes time $\Oh(n+m)$. Let $v_M$ be the root module of the decomposition tree $MD(G)$ with children $\{v_{M_1}, \ldots, v_{M_\ell}\}$, i.e., $P = \{M_1, \ldots, M_\ell\}$ is a modular partition of $G[M] = G$.
If $v_M$ is a series node or $s$ and $t$ are in a same module in $P$, we can compute $\lambda_G(s,t)$ in time $\Oh(n)$, see Lemma~\ref{lem:edgedisjointsamemodule}. Otherwise, we use the network defined in Definition~\ref{def:edgedisjointflownetwork}. Computing this network takes time $\Oh(\vert V(G') \vert +\vert E(G') \vert) = \Oh(\ell^2)$. Thus, we can compute $\lambda_G(s,t)$ in time $\Oh(\ell^3)$ using the maximum flow algorithm by Orlin~\cite{orlin2013max}. Since $\ell \leq \mw$ we have proven Theorem~\ref{thm:edgedisjointpaths}. 
Since the algorithm by Orlin takes time $\Oh(nm)$, another way to bound the running time of the computation of a maximum flow in $G'$ is $\Oh(\mw m)$, giving us the running time of $\Oh(\min \{\mw m+n, \mw^3+n+m\})$. 

\subparagraph{Kernel.}
The algorithm with running time $\Oh(\mw^3+n+m)$ can be easily modified to compute a kernel in linear time: We can compute in time $\Oh(n+m)$ an equivalent instance of a maximum flow problem of size $\Oh(\mw^2)$, which can be solved in time $\Oh(\mw^3)$. Such results were also achieved by Coudert et al.~\cite{CoudertDP18} for \prob{Eccentricities}, \prob{Hyperbolicity}, and \prob{Betweenness Centrality} parameterized by modular-width. It is easy to see that this holds in general for \emph{any} algorithm with running time $\Oh(f(k)+n+m)$: The running time is either dominated by $\Oh(f(k))$ or by $\Oh(n+m)$; we run the algorithm for $c \cdot (n+m)$ steps (for $c$ large enough), either it terminates or we can conclude that $f(k) \geq c \cdot (n+m)$ and our input graph is already a kernel of size $\Oh(f(k))$.\footnote{This can be generalized in an obvious way to running times of type $\Oh(f(k) + g(N))$, where $N$ denotes the input size.}

\subsection{Global Minimum Cut}

Now, we want to compute the global minimum (edge) cut for an unweighted graph $G = (V, E)$, i.e. $\lambda(G) = \min \{\lambda_G(s,t) \mid s, t \in V\}$. We can reduce the computation of a global (unweighted) minimum cut of $G$ to a single computation of a global weighted minimum cut in a graph closely related to the quotient graph of the root module. For this, we modify the algorithm for finding a minimum $s$-$t$ cut for fixed $s, t \in V$ seen in Section~\ref{sec:stedgecut}. 

\begin{theorem}\label{thm:globaledgecut}
 For every graph $G = (V, E)$ with modular-width $\mw$, \prob{Global Minimum Edge Cut} can be solved in time $\Oh({\mw^3} +n+m)$.
\end{theorem}

\subparagraph{Algorithm.}
Consider a graph $G = (V, E)$. We can assume that $G$ is connected, otherwise $\lambda(G) = 0$. First, we compute the modular decomposition tree $MD(G)$. Let $v_M$ be the root node of $MD(G)$. If $v_M$ is a series node it follows from Corollary~\ref{cor:edgedisjointseriesnode} that $\lambda_G(s,t) =\min \{\deg_G(s), \deg_G(t)\}$ for all pairs of vertices $s,t \in V$; therefore, $\lambda(G) = \min_{v \in V} \deg_G(v)$. Assume that $v_M$ is a prime node and let $\{v_{M_1}, \ldots, v_{M_\ell}\}$ be the children of $v_M$ in $MD(G)$, i.e., $P = \{M_1, \ldots, M_\ell\}$ is modular partition of $G[M] = G$. Let $(s^*,t^*) = \argmin \{\lambda_G(s,t) \mid s, t \in V\}$ and let $\delta(G) = \min_{v \in V} \deg_G(v)$. Obviously, $\lambda(G) \leq \delta(G)$. If $s^*$ and $t^*$ belonging to the same module $M_i \in P$ then $\lambda(G) = \delta(G)$ by Lemma~\ref{lem:edgedisjointsamemodule}. It is only possible that $\lambda(G) < \delta(G)$, if $s^*$ and $t^*$ are in different modules. 
The following lemma shows that that $s^*$ and $t^*$ are vertices of minimum degree in a module.
\begin{lemma}\label{lem:canpickstarvertices}
 Let $G = (V, E)$ be a graph and $P = \{M_1, \ldots, M_\ell\}$ be a modular partition of $G$. Let $(s^*, t^*) = \argmin \{\lambda_G(s,t) \mid s \in M_i, t \in M_j, i \neq j\}$. Then it is possible to pick $s^*$ and $t^*$ as vertices of minimum degree in their modules.
\end{lemma}

\begin{proof}
 As shown in Section~\ref{sec:stedgecut}, one can compute $\lambda_G(s,t)$ for $s \in M_i$ and $t \in M_j$ with $i \neq j$ by computing a maximum flow in the network $N$ defined in Definition~\ref{def:edgedisjointflownetwork}. The graph $G'$ will be the same for all $s \in M_i$ and $t \in M_j$, but the capacities on the edges $\{q_0, q_1\}$ amd $\{q_\ell, q_{\ell+1}\}$ change. These capacities are equal to $deg_{G[M_1]}(s)$, resp. $deg_{G[M_ell]}(t)$. Thus, they are minimal if we choose $s \in M_1$ and $t \in M_\ell$ such that $s$ and $t$ have minimum degree in $M_1$ resp. $M_\ell$.
\end{proof}
Now, we can create an auxiliary graph that is similar to graph $G'$ in Definition~\ref{def:edgedisjointflownetwork}, in order to compute $\lambda(G)$.

\begin{definition}\label{def:globalmincut}
 Let $G = (V, E)$ be a graph and let $P = \{M_1, ..., M_\ell\}$ be a modular partition of $G$ into $\ell \geq 2$ modules. Let $v_i \in M_i$ denote a vertex of minimum degree in $M_i$ for $i \in [\ell]$. 
 Let $G_{/P}$ be the quotient graph with vertex set $\{q_1, \ldots, q_\ell\}$.
 We define a weighted graph $G^*=(V^*,E^*)$ with weights $w\colon E^* \rightarrow \mathbb{N}$ as follows:
 \begin{itemize}
  \item The graph $G^*$ is initiated as being equal to $G_{/P}$.
  \item We add vertices $q_{i+\ell}$ to $V(G_{/P})$, each with the same neighbors as $q_i$ for $i \in [\ell]$.
  \item We add the edges $\{q_i, q_{i+\ell}\}$ to $E^*$.
  \item The $2 \ell$ vertices of $G^*$ correspond to the sets of vertices in the partition \\
  $P' = \{M'_1,M'_{\ell+1}, \ldots, M'_\ell, M'_{2\ell}\}$ with $M'_i = \{v_i\}$ and $M'_{i+\ell} = M_i \setminus \{v_i\}$ for $i \in [\ell]$. 
  \item The capacities on the edges of $G^*$ represent the number of edges between the corresponding vertex sets in $G$, i.e. \\ $c(q_i, q_j) = \vert M'_i \vert \vert M'_j \vert$ for $\{q_i, q_j\} \in E(G') \setminus \{\{q_k, q_{k+\ell}\}\mid k \in [\ell]\}$   with $i,j \in [2\ell]$ and \\
  $c(q_i, q_{i+\ell}) = deg_{G[M_i]}(v_i)$ for $i \in [\ell]$.
 \end{itemize}
\end{definition}

\begin{lemma}
 Let $G = (V,E)$ be a graph and $P = \{M_1, \ldots, M_\ell\}$ be a modular partition of $G$ into $\ell \geq 2$ modules. Let $(G^*,w)$ be the weighted auxiliary graph defined in Definition~\ref{def:globalmincut}. Then $\lambda(G) = \lambda((G^*, w))$. 
\end{lemma}

\begin{proof}
 Let $v_i$ be a vertex of minimum degree in $M_i$.
 We know, that $\lambda(G) = \min \{\lambda_G(s,t) \mid s,t \in V\}$ is the lowest maximum $s$-$t$ flow in $G$ considering every pair of vertices $s, t \in V$. By Lemma~\ref{lem:canpickstarvertices}, we only need to consider vertices of minimum degree in the modules. First, we observe that a maximum $v_i$-$v_j$ flow in $G$ (with $c \equiv 1)$ has the same value as a maximum $q_{i}$-$q_j$ flow in $G^*$ (with $c \equiv w)$ by the same argument as in the proof of Lemma~\ref{lem:primeflow}. Therefore, $\lambda(G) \geq \lambda((G^*, w))$. For the converse, we observe that $\lambda((G^*, w))$ corresponds to a maximum flow between two vertices $q_i$ and $q_j$ with $i,j \leq \ell$, because each pair $(q_i, q_{i+\ell})$ has the exact same neighborhood, but the capacities on incident edges of $q_i$ are smaller than the capacities on incident edges of $q_{i+\ell}$. Since every $q_i$-$q_j$ flow in $G^*$ with $i,j \leq \ell$ has the same value as a $v_i$-$v_j$ flow in $G$, we have proven the claim.
\end{proof}

\subparagraph{Running Time.}
Consider a graph $G = (V,E)$ with modular-width $\mw$. Computing the modular decomposition tree $MD(G)$ takes linear time $\Oh(n + m)$. Let $v_M$ be the root node of $MD(G)$ with children $\{v_{M_1}, \ldots, v_{M_\ell}\}$, i.e, $P = \{M_1, \ldots, M_\ell\}$ is a modular partition of $G[M] = G$. If $v_M$ is a parallel node or a series node, we can compute $\lambda(G)$ in time $\Oh(n)$. Consider $v_M$ to be a prime node. Generating the instance $(G^*,w)$ according to Definition~\ref{def:globalmincut} takes time $\Oh(\vert V(G^*) \vert + \vert E(G^*)\vert) = \Oh(\ell^2)$. Note that $\vert V(G^*) \vert = 2 \ell$. Computing a weighted global minimum edge cut in the undirected graph $G^*$ can be done in time $\Oh(\ell^3)$ by using the algorithm of Stoer and Wagner~\cite{stoer1997simple}. Since $\ell \leq \mw$ we have proven Theorem~\ref{thm:globaledgecut}.

\section{Vertex-Disjoint Paths}\label{sec:vertexdisjointpaths}

A connected graph $G = (V, E)$ is said to be \emph{k-vertex-connected} if one can delete up to $k$ arbitrary vertices and $G$ stays connected. The \emph{vertex-connectivity} of $G$, denoted by $\kappa(G)$, is the largest $k$ for which $G$ is $k$-vertex-connected. In other words, $\kappa(G) = \min\{\kappa(s,t) | s,t \in V, \{s,t\}\notin E\}$ and $\kappa(s,t)$ denotes the minimum size of an $s$-$t$-vertex separator. By Menger's Theorem \cite{menger1927allgemeinen}, for $s,t \in V$, the minimum size of an $s$-$t$-vertex separator  is exactly the size of the maximum number of vertex-disjoint $s$-$t$-paths, denoted by $\Pi(s,t)$. 
The latter one is as well of independent interest. One can compute $\Pi(s,t)$ by solving an $s$-$t$ vertex capacitated flow with capacities equal to one for every vertex. Instead of focusing on $\Pi(s,t)$, we directly solve the more general problem of computing a maximum $s$-$t$ vertex flow with an arbitrary capacity function $c \colon V\setminus \{s,t\} \rightarrow \mathbb{R^+}$. 
This will be the focus of Section~\ref{sec:stflow}. In Section~\ref{sec:globalminvertexcut} we will then focus on computing $\kappa(G)$, but again in the more general setting with vertex capacities.

\subsection{Maximum s-t Vertex Flow} \label{sec:stflow}

\begin{theorem}\label{thm:stvertexflow}
 For every graph $G = (V, E)$ with modular-width $\mw$, we can compute a maximum $s$-$t$ vertex-capacitated flow in time $\Oh({\mw^3}+n+m)$.
\end{theorem}

\subparagraph{Algorithm.}
Consider a network $ N = (G,s,t,c)$ consisting of a graph $G = (V, E)$ with $s, t \in V$ and a capacity function $c \colon V \rightarrow \mathbb{R^+}$. We want to compute a maximum $s$-$t$ flow in $N$. Assume that $\{s,t\} \notin E$, otherwise the maximum $s$-$t$ flow is unbounded. Our algorithm will traverse the modular decomposition tree $MD(G)$ in a bottom-up manner.
Let $v_M$ be the node in $MD(G)$ that corresponds to the smallest module $M$ with $s, t \in M$ (that is the lowest common ancestor of the two leaves in $MD(G)$ corresponding to $\{s\}$ and $\{t\}$). We distinguish three cases to compute a maximum $s$-$t$ flow in $G[M]$. If $v_M$ is a parallel node, the maximum $s$-$t$ flow in $G[M]$ is zero. If $v_M$ is a series node, $s$ and $t$ are adjacent (which we ruled out). For the case of $v_M$ being a prime node in $MD(G)$, let $\{v_{M_1}, .., v_{M_\ell}\}$ be the set of children of $v_M$. This means that $P = \{M_1, .., M_\ell\}$ is a modular partition of $G[M]$. Let $G[M]_{/P}$ be the quotient graph of $G[M]$ with the vertex set $\{q_1, \ldots, q_\ell\}$. We assume, w.l.o.g.~, that $s \in M_1$ and $t \in M_\ell$. 
The following lemma shows that the maximum $s$-$t$ flow in $(G[M], s, t, c)$ can be computed as the maximum $q_1$-$q_\ell$ flow in $G[M]_{/P}$ where the capacity of $q_i$ is simply the sum of capacities of the vertices of its corresponding module.

\begin{lemma}\label{lem:stflow}
 Let $N = (G, s, t, c)$ be a flow network with a graph $G = (V, E)$, vertices $s, t \in V$, and $c \colon V \rightarrow \mathbb{R}^+$. Let $P = \{M_1, ..., M_\ell\} $ be a partition of $V$ into modules with $s \in M_1$ and $t \in M_\ell$. 
 Then, the maximum flow in $N$ is equal to the maximum flow in $N' = (G_{/P}, q_1, q_\ell, c')$ with the capacity function $c' \colon V(G_{/P}) \rightarrow \mathbb{R}^+$ defined by $c'(q_i) = \sum_{v \in M_i} c(v)$.
\end{lemma}

\begin{proof}

 For a graph $G = (V, E)$ we denote by $\overleftrightarrow{G} = (V, \overleftrightarrow{E})$ the directed graph with $\overleftrightarrow{E} = \{(u,v),(v,u) \mid \{u,v\} \in E \}$.
 Let $f' \colon E(\overleftrightarrow{G_{/P}}) \rightarrow \mathbb{R}^+$ be a maximum flow in $N'$. Then, $f'$ corresponds to a flow $f$ in $N$ not using any edges inside the modules.
 Therefore, the size of maximum flow in $N$ is at least the size of maximum flow in $N'$.
 
 Conversely, consider a maximum flow $f \colon E(\overleftrightarrow{G}) \rightarrow \mathbb{R}^+$ of $N$. 
 We show that there is always a maximum flow in $N$ that does not use edges \textit{inside} a module. To do this, consider the potential function $\psi = \sum_{M \in P} \sum_{e \in E(\overleftrightarrow{G[M]})} f(e)$. 
 If $\psi = 0$ then $f$ has the desired property, so assume $\psi > 0$. 
 It is well known that a flow in a graph can be decomposed into flows along paths in the graph. Thus, in any decomposition of $f$ there must be an $s$-$t$ path $Q=(s=v_1,v_2,\ldots,v_r=t)$ where some $v_i$ and $v_{i+1}$, with $i\in[r-1]$, are contained in the same  module $M$. We can replace flow along $Q$ by sending the same amount of flow along any minimal subpath $Q'$ of $Q$, which can be obtained by repeatedly shortcutting $Q$ along edges between vertices that are not consecutive. Clearly, this does not affect used capacity at vertices in $Q'$ and only decreases used capacity at vertices that are in $Q$ but not $Q'$. It is easy to see that the shortcutting operation keeps at most one of $v_i$ and $v_{i+1}$ in $Q'$ (that is, at most one vertex from the subpath of $Q$ in $M$ is retained).
 
 Applying the above method iteratively, we have modified the flow $f$ such that $\psi = 0$. A flow $f$ in $N$ not using edges inside modules directly corresponds to a flow $f'$ in $N'$ with $\vert f' \vert = \vert f \vert$.  
\end{proof}

After computing the maximum $s$-$t$ flow in $G[M]$, we can directly compute the maximum $s$-$t$ flow in $G$:

\begin{lemma}\label{lem:furtheraugmentation}
 Let $N = (G, s, t, c)$ be a network with a graph $G = (V, E)$, vertices $s, t \in V$, and a capacity function $c \colon V \rightarrow \mathbb{R}^+$ and let $M \subseteq V$ be a module of $G$ with $s,t \in M$. Let $f_M$ be a maximum $s$-$t$ flow in $G[M]$. Then, the maximum flow value in $N$, denoted by $MF(N)$ is: 
 \begin{align*}
  MF(N) = \vert f_M \vert + \sum_{v \in N_G(s) \setminus M} c(v)
 \end{align*}
 
\end{lemma}

\begin{proof}
For every $v \in N_G(s)\setminus M$, $v$ is also a neighbor of $t$, such that we can augment the flow $f_M$ for every $v \in N_G(s)\setminus M$ by the augmenting path $(s, v, t)$. Afterwards, every further augmenting path would be an augmenting path completely in $G[M]$, contradicting the maximality of $f_M$.  
\end{proof}

\subparagraph{Running Time.}
Consider a graph $G = (V, E)$ with modular-width $\mw$ and a network $N = (G, s, t, c)$ with $s, t \in V$ and $c\colon V \rightarrow \mathbb{R}^+$. Computing the modular decomposition tree $MD(G)$ takes time $\Oh(n + m)$. Finding the node $v_M$ in $MD(G)$ that corresponds to the smallest module $M$ with $s,t \in M$ (i.e. the lowest common ancestor of the two leaf nodes in $MD(G)$ that corresponds to the graph only consisting of vertex $s$ resp. $t$) takes time $\Oh(n)$. We can compute the size of a maximum $s$-$t$ flow in $G[M]$ by
either concluding that the flow is equal to zero (if $M$ is parallel) or by
using Lemma~\ref{lem:stflow} (if $M$ is prime). The latter can be done in time $\Oh(\ell^3)$ using the algorithm due to Orlin~\cite{orlin2013max} where $\ell \leq \mw$ denotes the size of the quotient graph of $M$. Note, that $M$ cannot be series, since we assume $\{s,t\} \notin E$. Due to Lemma~\ref{lem:furtheraugmentation}, we can compute the maximum flow in $N$ in additional $\Oh(n)$ time, which proves Theorem~\ref{thm:stvertexflow}. Again, one can also bound the computation of the maximum flow in $G_{/P}$ by $\Oh(\mw m)$, giving us the running time of $\Oh(\min \{\mw m+n, \mw^3+n+m\})$.

\subsection{Global Minimum Vertex Cut} \label{sec:globalminvertexcut}

In the \prob{Global Minimum Vertex Cut} problem we are given an input graph $G = (V, E)$ and a capacity function $c \colon V \rightarrow \mathbb{R}^+$ and need to compute a set $X \subset V$ of minimum capacity such that $G - X$ is disconnected. This is equivalent to finding a pair of vertices $s,t\in V$ such that the maximum $s$-$t$ vertex flow is minimized. 
An $s$-$t$ vertex cut is a set $X \subsetneq V \setminus \{s,t\}$ such that $s$ and $t$ are disconnected in $G - X$. We denote the minimum $s$-$t$ vertex cut, respectively the maximum $s$-$t$ flow in $G$, with capacity function $c$ by $\Pi_{(G,c)}(s,t)$ and omit $c$ if the capacity function is clear. We denote by $\Pi((G,c))$ the global minimum vertex cut, thus $\Pi((G,c))  = \min_{s,t \in V(G)}\Pi_{(G,c)}(s,t)$ and again omit $c$ if the capacity function is clear. We set $\Pi_{(G,c)}(s,t) = \infty$, if $s$ and $t$ are adjacent.

\begin{theorem}\label{thm:globalminvertexcut}
 For every graph $G = (V, E)$ with modular-width at most $\mw$, \prob{Global Minimum Vertex Cut} can be solved in time $\Oh({n \mw^2 \log \mw} + m)$.
\end{theorem}

\subparagraph{Algorithm.}
Consider a graph $G = (V, E)$ and a capacity function $c \colon V \rightarrow \mathbb{R}^+$.
First, we compute the modular decomposition tree $MD(G)$. 
For every node $v_M$ in $MD(G)$, we will compute the total capacity of the corresponding module $c(M) := \sum_{v \in M} c(v)$ and the size of a minimum vertex cut $\Pi(G[M])$. 
Then, for the root node $v_M$ in $MD(G)$, we will compute $\Pi(G[M]) = \Pi(G)$. 
We traverse the decomposition tree in a bottom-up manner. 
We set $\Pi(G[M]) = \infty$ for all leaf modules $v_M$. 
In the following, for computing $\Pi(G[M])$ for any module $M$ corresponding to a node $v_M$ in $MD(G)$, we can always assume that we have already computed the size of a minimum vertex cut for all child nodes. 
The next lemma shows that this information is enough to compute the minimum vertex cut in the parent node.

\begin{lemma}
 Let $G = (V, E)$ and $M$ be a module of $G$. Let $s, t, u , v \in M$. If $\Pi_{G[M]}(s,t) \leq \Pi_{G[M]}(u,v)$ then also $\Pi_{G}(s,t) \leq \Pi_{G}(u,v)$.
\end{lemma}

\begin{proof}
 By Menger's Theorem, $\Pi_{G[M]}(s,t)$ and $\Pi_{G[M]}(u,v)$ both correspond to a maximum vertex flow in $G[M]$. Due to Lemma~\ref{lem:furtheraugmentation}, the maximum flow in $G$ increases by the same amount for both values.
\end{proof}

Let $v_M$ be a node in  $MD(G)$ with children $\{v_{M_1}, \ldots, v_{M_\ell}\}$. This means $P = \{M_1, \ldots, M_\ell\}$ is a modular partition of $G[M]$. If $v_M$ is a parallel node, it holds that $\Pi(G[M]) = 0$ and that $c(M) = \sum_{i \in [\ell]} c(M_i)$. Now, assume that $v_M$ is a series node. Again we compute $c(M)= \sum_{i \in [\ell]} c(M_i)$. To compute $\Pi(G[M])$, we first observe that for all $s \in M_i$, $t \in M_j$ with $i \neq j$ we have $\Pi_G[M](s,t) = \infty$, because they are adjacent. Therefore, a minimum vertex cut in $G[M]$ has to be an $s$-$t$ vertex cut with $s$ and $t$ being in the same module $M_i$. Hence, we can compute $\Pi(G[M])$ by extending every minimum cut in a module by the summation of the capacities of the neighboring modules and taking the minimum:
\begin{align*}
 \Pi(G[M]) &= \min_{i \in [\ell]} \{ \Pi(G[M_i])+ \sum_{k \in [\ell]\setminus\{i\}}c(M_k) \} \\
	  &= \min_{i \in [\ell]} \{ \Pi(G[M_i]) + c(M) - c(M_i) \} 
\end{align*}
Finally, assume that $v_M$ is a prime node. First, we compute $c(M)$ as before. There are two different types of vertex cuts in $G[M]$. The first type of vertex cut is an $s$-$t$ cut for vertices $s$ and $t$ in a same module. In this case, every vertex cut in $G[M_i]$ has to be extended to a vertex cut in $G[M]$ by adding the capacities of neighboring modules, formally we have
\begin{align*}
 \hat{\Pi}(G[M]) = \min_{i \in [t]} \left\{ \Pi(G[M_i]) + \sum_{M_j : \{M_i,M_j\} \in E(G[M]_{/P}) } c(M_j)  \right\} 
\end{align*}
The second type of vertex cuts in $G[M]$ is an $s$-$t$ vertex cut with $s$ and $t$ being in different (and non-adjacent) modules. Lemma~\ref{lem:quotientcut} shows that in this case we can compute the maximum vertex cut in $G[M]$ by computing a maximum vertex cut in $G[M]_{/P}$ with a capacity function $c' \colon V(G[M]_{/P}) \rightarrow \mathbb{R}^+$ defined by $c'(q_i) = c(M_i)$.

\begin{lemma} \label{lem:quotientcut}
 Let $G = (V, E)$ be a graph with a modular partition $P$ and let $s, t \in V$. Let $M_s, M_t \in P$ with $M_s \neq M_t$, $s \in M_s $ and $t \in M_t$. Let $X \subseteq V$ be a minimum $s$-$t$ vertex cut in $G$. Then, for every module $M \in P$, either $M \subseteq X$ or $M \cap X = \emptyset$.
\end{lemma}

\begin{proof}
 Assume for contradiction that there is a minimum $s$-$t$ vertex cut $X$ in $G$ and that there is a module $M \in P$ with $M \cap X = X_M$ and $\emptyset \subsetneq X_M \subsetneq M$. We claim that in this case $X' = X \setminus X_M$ is an $s$-$t$ vertex cut in $G$, contradicting the minimality of $X$.

 Suppose $X'$ is not an $s$-$t$ vertex cut in $G$. Consider any $s$-$t$ path $Q$ that contains at most one vertex in each module. (Such a path exists by elementary properties of modules.) There must be a vertex $p\in X_M$ in $Q$ as $X=X'\cup X_M$ is an $s$-$t$ vertex cut. Clearly, replacing $p$ by any vertex $q\in M\setminus X_M\neq\emptyset$ yields another $s$-$t$ path that also avoids $X$, a contradiction.

 \end{proof}

After computing the global minimum cut in the quotient graph with capacity function $c'$, we simply need to compare this value with $\hat{\Pi}(G[M])$ and choose the smaller value as the minimum cut for $G[M]$.

\subparagraph{Running Time.}
Let $G = (V, E)$ be a graph with capacity function $c \colon V \rightarrow \mathbb{R}^+$ that is an instance for the \prob{Global Minimum Vertex Cut} problem and let $\mw$ be the modular-width of $G$. The modular decomposition tree $MD(G)$  
can be computed in linear time. For every node $v_M$ in $MD(G)$, computing $c(M)$ takes total time $\Oh(n)$, since it will be iteratively computed from the values of the child nodes and $MD(G)$ has less than $2n -1$ nodes. By the same argument the total running time for all series nodes is $\Oh(n)$. Let $v_M$ be a prime node in $MD(G)$ with children $\{v_{M_1}, \ldots, v_{M_\ell}\}$, i.e., $P=\{M_1, \ldots, M_\ell\}$ is a modular partition of $G[M]$. Note, that $\ell \leq \mw$. To compute $\hat{\Pi}(G[M])$ we need time $\Oh(\ell^2)$. We can find a global minimum vertex capacitated cut in $G[M]_{/P}$ by solving a global edge capacitated cut in a directed graph using standard reductions between flow/cut variants. The latter can be solved in time $\Oh(\ell^3 \log \ell)$ due to Hao and Orlin~\cite{hao1994faster}. By a similar careful summation as done in Section~\ref{sec:maximummatching} we obtain a total running time of $\Oh(n  {\mw^2 \log \mw} + m)$, which proves Theorem~\ref{sec:globalminvertexcut}. By a different summation (similar as in Section~\ref{sec:maximummatching}) one can also obtain the running time $\Oh(m {\mw \log \mw}  +n)$. This leads to an $\Oh( \min \{m {\mw \log \mw}  +n, n  {\mw^2 \log \mw} + m\})$ time algorithm for \prob{Global Minimum Vertex Cut}.

\section{Conclusion}\label{sec:conclusion}

We have obtained efficient parameterized algorithms for the problems \prob{Maximum Matching}, \prob{Maximum $b$-Matching}, \prob{Triangle Counting}, and several path- and flow-type problems with respect to the modular-width $\mw$ of the input graph. All time bounds are of form $\Oh(f(\mw)n+m)$, $\Oh(n+f(\mw)m)$, or $\Oh(f(\mw)+n+m)$, where the latter can be easily seen to imply linear-time preprocessing to size $\Oh(f(\mw))$. 
Throughout, the dependence $f(\mw)$ is very low and several algorithms are adaptive in the sense that their time bound interpolates smoothly between $\Oh(n+m)$ when $\mw=\Oh(1)$ and the best known unparameterized running time when $\mw=\Theta(n)$. 
Thus, even if typical inputs may have modular width $\Theta(n)$ (a caveat that all structural parameters face to some degree), using these algorithms costs only a constant-factor overhead and already $\mw=o(n)$ yields an improvement over the unparameterized case.

As mentioned in the introduction, (low) modular-width seems useless in problems where edges are associated with weights and/or capacities. Intuitively, these numerical values distinguish edges between adjacent modules $M$ and $M'$, which could otherwise be treated as largely equivalent. For concreteness, consider an instance $(G,s,t,w)$ of the \prob{Shortest $s$,$t$-Path} problem where $w\colon E(G)\to\mathbb{N}$ are the edge weights. Clearly, the distance from $s$ to $t$ is unaffected if we add the missing edges of $G$ and let their weight exceed the sum of weights in $w$. However, the obtained graph is a clique and has constant modular-width. Similar arguments work for other edge-weighted/capacitated problems like \prob{Maximum Flow} using either huge or negligible weights. In each case, running times of form $\Oh(f(\mw)g(n,m))$ would imply time $\Oh(g(n,m))$ for the unparameterized case (without considering modular-width), so the best such running times cannot be outperformed even for low modular-width.

Apart from developing further efficient (and adaptive?) parameterized algorithms relative to modular-width there are other directions of future work. Akin to conditional lower bounds via fine-grained analysis of algorithms it would be interesting to prove optimality of efficient parameterized algorithms for all regimes of the parameters (e.g., like Bringmann and K\"unnemann~\cite{BringmannK18}). Which other (graph) parameters allow for adaptive parameterized running times so that even nontrivial upper bounds on the parameter imply faster algorithms than the unparameterized worst case?

\end{document}